\documentclass[conference]{IEEEtran}
\IEEEoverridecommandlockouts
\usepackage{cite}
\usepackage{amsmath,amssymb,amsfonts}
\usepackage{algorithmic}
\usepackage{graphicx}
\usepackage{textcomp}
\usepackage{xcolor}

\newtheorem{theorem}{Theorem}[section]

\newtheorem{lemma}[theorem]{Lemma}

\def\BibTeX{{\rm B\kern-.05em{\sc i\kern-.025em b}\kern-.08em
    T\kern-.1667em\lower.7ex\hbox{E}\kern-.125emX}}

\begin{document}

\title{Robust Current Regulation of MMC-based MTDC Power Systems based on Lyapunov Inequality}

\author{\IEEEauthorblockN{Victor Daniel Reyes Dreke}
\IEEEauthorblockA{\textit{IEPG Group, TU Delft} \\
Delft, Netherlands \\
v.d.reyesdreke@tudelft.nl}
\and
\IEEEauthorblockN{Rahul Rane}
\IEEEauthorblockA{\textit{IEPG Group, TU Delft} \\
Delft, Netherlands \\
r.m.rane@tudelft.nl}
\and
\IEEEauthorblockN{Aleksandra Leki\'{c}}
\IEEEauthorblockA{\textit{IEPG Group, TU Delft} \\
Delft, Netherlands \\
a.lekic@tudelft.nl}
\thanks{The work is supported by Horizon Europe project PROSECCO, under grant agreement 101160687, and by NWO Veni project SAFE-GRID, with project number 20248.}}

\maketitle

\begin{abstract}
Multi-terminal DC (MTDC) transmission systems based on modular multilevel converters (MMCs) are a key component of the envisioned future energy sector, where sustainability and efficiency are increasingly prioritized. 
To ensure their reliable operation, MMC currents must be regulated safely and rapidly under a wide range of uncertain operating conditions.
Consequently, the design of current controllers faces a fundamental challenge: achieving fast transient response while maintaining robustness against uncertainties. 
This paper addresses this challenge by proposing a linear matrix inequality (LMI)–based design framework that leverages Lyapunov stability conditions to synthesize a less conservative static state-feedback controller. 
The proposed design method explicitly accounts for system constraints, including input saturation and overcurrent limits. 
The proposed method's effectiveness is assessed on the CIGRE MT-HVDC benchmark, simulated in RTDS\textsuperscript{\textregistered}, and compared with existing methods.
    
\end{abstract}


\section{Introduction}
\label{lab:Introduction}

Multi-terminal high-voltage DC (MTDC) grids are becoming the core of the future energy transmission landscape. 
MTDC systems enable high-power energy transfer with improved efficiency, thereby facilitating the integration and sharing of distributed energy resources such as offshore wind farms and large-scale solar plants in a more sustainable manner \cite{Shetgaonkar2023_MPCHVDC}. 
In these cases, modular multilevel converters (MMCs) serve as the primary power electronic device, responsible for efficient energy conversion and reliable AC/DC power transmission \cite{Beddard2015_MMCModelling}.

To ensure the reliable operation of MTDC systems, MMC currents must be regulated accurately and safely over a broad range of operating conditions \cite{Ansari2020_MMCSurvey}. These conditions may vary rapidly and unpredictably due to fluctuations in power demand, network contingencies, or external disturbances, including natural phenomena \cite{Kumar2024_MMCFaults,Tavakoli2021_MMCACFaults}. 
MMC current dynamics are inherently nonlinear and multivariable \cite{Dekka2019_MPCMMCSurvey}. The simplified models commonly used for controller synthesis and analysis are therefore subject to both structured and unstructured uncertainties, coming from parameter variations, passive component tolerances, and modeling approximations such as linearization and discretization \cite{Louybary2024_MMC_LMI_Uncertainties}. Consequently, the controller design must resolve a fundamental control-theoretic trade-off between fast transient performance and robustness-induced conservativeness.

Existing MMC current control strategies typically treat transient performance and robustness separately. Fast dynamics are often achieved through heuristic tuning of controller hyperparameters, whereas robustness is enforced via synthesis methods such as $\mathcal{H}_\infty$ control~\cite{Gil2017_RobustHVDC,Tavakoli2021_OptimalMMC} and $\mu$-synthesis~\cite{tavakoli2022_RobustHVDC}. 
Although these approaches are well established, the resulting controllers may be conservative due to worst-case optimization during synthesis.
LMI-based robust controllers~\cite{Ayari2017_RobustHVDC,Belhaouane2019_RobustHVDC} provide an alternative by synthesizing static state-feedback gains that guarantee constraint-admissible stability and optimality. 
These methods offer simpler implementation and strong robustness properties; however, their conservativeness may limit transient performance.
In general, the transient response achieved in~\cite{Gil2017_RobustHVDC}-\cite{Belhaouane2019_RobustHVDC} can be improved through careful tuning of weighting filters and matrices. 
However, this process is typically time-consuming and does not systematically ensure the desired performance.

This paper addresses this challenge by proposing an LMI-based design framework that leverages Lyapunov inequalities to synthesize a less conservative static state-feedback gain. 
The proposed design emphasizes fast state regulation by formulating a synthesis problem that maximizes a constraint-admissible domain of attraction, without explicitly penalizing the control input. 
Robustness is ensured through Lyapunov-based constraints that guarantee stability for all admissible uncertainty realizations while maintaining system trajectories close to the nominal response. 
By explicitly accounting for user-defined system constraints, the resulting robust current regulator avoids excessively aggressive control actions. 
The proposed method's effectiveness is assessed on the CIGRE MT-HVDC benchmark, simulated in RTDS\textsuperscript{\textregistered}, and compared with existing methods.

The paper is organized as follows. Section II provides modelling of the MMC. In Section III, a robust current regulator design is described. Section IV gives results from numerical simulations, while Section V presents concluding remarks.

\section{MMC Fundamentals}
\label{sec:Fundamentals}
MMCs are voltage source converters (VSCs) comprised of 3 legs with two arms in which $N$ submodules  
\textcolor{black}{(SMs)} 
are connected in series to an equivalent resistor ($R_m$) and inductor ($L_m$), see Figure~\ref{fig:MMC_SEST26_}. 
Each submodules constitute a modular AC/DC power conversion unit driven by switching signals.
The arms are named as upper and lower arm, denoted by ${n \in \{u,l\}}$, and the phases are denoted by $m \in \{a,b,c\}$. 
\begin{figure}[htp]
    \centering
    \includegraphics[width=\linewidth]{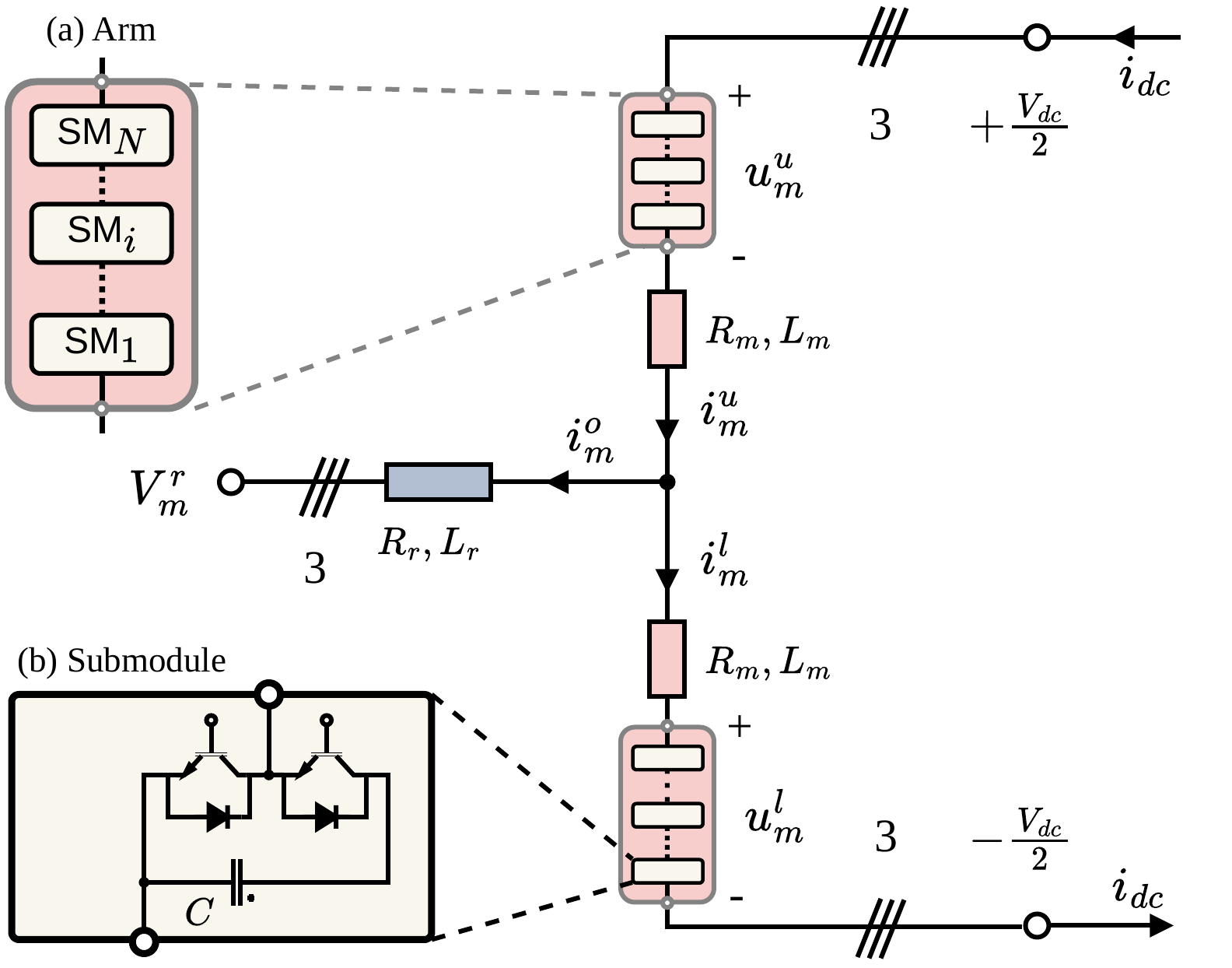}
    \caption{Simplified diagram of the MMC topology  per-phase $m \in \{a,b,c\}$. (a) MMC arms. (b) Half-bridge submodule.}
    \label{fig:MMC_SEST26_}
\end{figure}

\subsection{MMC Current Dynamics}
The MMC current dynamics can be analyzed in different manners, see \cite{Bergna2018_Generalized_MMC_MOdel}. 
We opt to use an average equivalent representation, where the following simplifications are considered: (\emph{i}) the capacitor voltage in each submodule is balanced at all time and (\emph{ii}) the switching signals are driven by an aggregated continuous signal, which are the modulation indices. 

Kirchhoff's voltage law shows that MMC currents are governed by transformer 
\textcolor{black}{(TF)} 
$(V^r_{m})$, DC-link $(V_{dc})$, and arm $(u^{u,l}_{m})$ voltages, such that
\begin{subequations}
\label{eq:DifferentialEquationUpper}
    \begin{align}
        \label{eq:DifferentialEquationUpper1}
             L_{m} \frac{\mathrm{d}\, {i}_{m}^{{u}}}{\mathrm{d} t}  &=-R_{m} {i}_{m}^{{u}}+u^u_{m}+V^r_{m}+  L_{r} \frac{\mathrm{d}\, {i}_{m}^{o}}{\mathrm{d} t}+R_{r}{i}_{m}^{o}-\frac{V_{dc}}{2}, \\
             \label{eq:DifferentialEquationUpper2}
             L_{m} \frac{\mathrm{d} \, {i}_{m}^{l}}{\mathrm{d} t}  &=-R_{m} {i}_{m}^{l}+u^l_{m}-V^r_{m}  -L_{r}\frac{\mathrm{d}\, {i}_{m}^{o}}{\mathrm{d} t}-R_{r}{i}_{m}^{o}-\frac{V_{dc}}{2},
    \end{align}
\end{subequations}
where $L_{m}$ and $L_{r}$ are the arm and transformer inductance, respectively; and $R_{m}$ and $R_{r}$ are the arm and transformer resistance, respectively. 
The values of the MMC's passive components include both nominal and tolerance quantities. 
Therefore, these values contain nominal and uncertain quantities, where 
\begin{equation}
    \begin{aligned}
        L_{m}&= L_{m0} +  L_{m\Delta}, \quad \, R_{m} = R_{m0} +  R_{m\Delta},\\ 
        L_{r}&= L_{r0} + L_{r \Delta}, \quad\quad  R_{r}= R_{r0} + R_{r\Delta}.
    \end{aligned}
\end{equation}

    To distinguish between nominal and uncertain quantities, subscript notations  ${(\cdot)}_{\cdot 0}$ and ${(\cdot)}_{\cdot \Delta}$ are used, respectively.  
    $ L_{m \Delta}$, $ R_{m \Delta}$, $ L_{r \Delta}$, $ R_{r \Delta}$, 
    represent structured uncertainties, these coefficients are defined as bounded parameters, such that
    \begin{equation}
    \label{eq:uncertain_tolerance_coefficeint}
        \left( L_{m \Delta}, R_{m \Delta}, L_{r \Delta}, R_{r \Delta} \right) \in \mathcal{X}_{\Delta} \subset \mathbb{R}.
    \end{equation}
The MMC currents determine the power transfer through the output and circulating currents, defined per phase \( m \) as 
\begin{equation}
\label{eq:outputCurrentsDef}
    i^{o}_{m} = i_m^u - i_m^l \quad \text{and} \quad i^{c}_{m} = \tfrac{1}{2} (i_m^u + i_m^l),
\end{equation}
respectively. 
From \eqref{eq:DifferentialEquationUpper}, the differential equations describing the dynamics of \( i_o \) and \( i_c \) are derived. 
To facilitate controller design and analysis, these dynamics are transformed to the \(dqz\) reference frame using the Park transformation \cite{Bergna2018_Generalized_MMC_MOdel}. 
In the next section, we present the corresponding linearized state-space realization of \( i_o \) and \( i_c \) in the \(dqz\) reference frame.

\subsection{MMC State-Space Model}
Following \eqref{eq:DifferentialEquationUpper} and modeling procedures in \cite{Shetgaonkar2023_MPCHVDC}, we represent the dynamics of \( i_o \) and \( i_c \) in the $dqz$-frame as uncertain discrete state-space models~\footnotemark{} with sampling time $T_s$ as follows:
\footnotetext{Notation $x^{+}$ describes the next sample, such that $x^{+} \triangleq x(k+1)$. }
\begin{itemize}
    \item For the output current 
    \textcolor{black}{control (OCC),} 
    we define their state and control input vectors as 
    \begin{equation}    
        x_{odq} =\begin{bmatrix} i^o_d & i^o_q\end{bmatrix}^{\top}, \; \text{and}\; u_{odq} = \begin{bmatrix}u^o_d & u^o_q\end{bmatrix}^{\top},
    \end{equation}
    respectively. Then, its dynamics yield 
    \begin{equation}
    \label{eq:ss_output_current_dq}
        x^+_{odq} = \left(\textbf{A}^o_0 + \textbf{A}_{\Delta}\right)x_{odq}+ \left(\textbf{B}^o_0 + \textbf{B}_{\Delta}\right)u_{odq}
    \end{equation}
     where $\textbf{A}_0 := e^{\textbf{A}^o_c{T}_{s}}$ and $\textbf{B}_0:= \int_0^{T_s}  e^{\textbf{A}^o_{c}\tau} \textbf{B}^o_{c} d\tau$ are the nominal state-space matrices with 
    \begin{equation*}
       \begin{aligned}
           \textbf{A}^o_c = \begin{bmatrix}
        -\frac{R_{eq}^{0}}{L_{eq}^{0}} & -\omega_{0} \\
            \omega_{0} & -\frac{R_{eq}^{0}}{L_{eq}^{0}}
        \end{bmatrix}, \,  \text{and} \,   \textbf{B}^o_c = \begin{bmatrix}
            \frac{1}{L_{eq}^{0}} & 0\\ 0 & \frac{1}{L_{eq}^{0}} 
        \end{bmatrix},
       \end{aligned} 
    \end{equation*}
    $L_{eq}^{0} = L_{r0}+ \frac{L_{m0}}{2},$ $R_{eq}^{0} = R_{r0}+ \frac{R_{m0}}{2}$, $\omega_0 = 2\pi f_0$.
    \item For the circulating current 
    \textcolor{black}{control (CCC),} 
    we define its state and control input vectors as 
    \begin{equation}
        x_{cdq} =\begin{bmatrix} i^c_d & i^c_q\end{bmatrix}^{\top}, \; \text{and}\; u_{cdq} = \begin{bmatrix}u^c_d & u^c_q\end{bmatrix}^{\top},
    \end{equation}
    respectively. Then, its dynamics yield 
    \begin{equation}
    \label{eq:ss_circulating_current_dq}
       x^+_{cdq} = \left(\textbf{A}^c_0 + \textbf{A}_{\Delta}\right)x_{cdq}+ \left(\textbf{B}^c_0 + \textbf{B}_{\Delta}\right)u_{cdq}
    \end{equation}
    where $\textbf{A}_0 := e^{\textbf{A}^c_c{T}_{s}}$ and $\textbf{B}_0:= \int_0^{T_s}  e^{\textbf{A}^c_{c}\tau} \textbf{B}^c_{c} d\tau$ are the nominal state-space matrices with 
    \begin{equation*}
       \begin{aligned}
           \textbf{A}^c_c = \begin{bmatrix}
        -\frac{R_{m0}}{L_{m0}} & 2\omega_0 \\
            -2\omega_0 & -\frac{R_{m0}}{L_{m0}}
        \end{bmatrix}, \,  \text{and} \, \textbf{B}^c_c = \begin{bmatrix}
            -\frac{1}{L_{m0}} & 0\\ 0 & -\frac{1}{L_{m0}} 
        \end{bmatrix}.
       \end{aligned} 
    \end{equation*}
\end{itemize}

In both cases, $(\textbf{A}_{\Delta },  \textbf{B}_{\Delta})$ are matrices filled with unknown uncertainties coefficients, i.e.,
\begin{equation}
    \rho = [\rho^{a}_{1,1}, \ldots, \rho^{b}_{n_x,n_u}]^{\top}\in \mathbb{P} \subset \mathbb{R}^{n_{\rho}},
\end{equation}
such that
\begin{equation}
\label{eq:coefficient_unknown}
  \textbf{A}_{\Delta } {=} \begin{bmatrix} \rho^{a}_{1,1} &{\ldots}&\rho^{a}_{1,n_x}\\ \vdots &\ddots & \vdots\\ \rho^{a}_{n_x,1} &{\ldots}&\rho^{a}_{n,n}\end{bmatrix},\;\textbf{B}_{\Delta}  {=} \begin{bmatrix}\rho^{b}_{1,1} &{\ldots}&\rho^{b}_{1,n_u}\\ \vdots &\ddots & \vdots\\ \rho^{b}_{n_x,1} &{\ldots}&\rho^{b}_{n_x,n_u}\end{bmatrix}.
\end{equation}
where $n_{\rho}={n_x}^2+ {n_x}\,{n_u}$, ${n_x}$ is the state dimension and ${n_u}$ input dimension. 
We abuse $\textbf{A}_{\Delta },  \textbf{B}_{\Delta}$ notation, as they are not necessarily the same for both output and circulating currents.  

Additionally, we derive the augmented models in \eqref{eq:ss_output_current_dq} and \eqref{eq:ss_circulating_current_dq} to enforce zero-error tracking when following step references. 
Discrete-time derivation is denoted by the operator $\delta \cdot$, such that 
\begin{equation}
    \delta x(k):= x(k) -x(k-1). 
\end{equation}
For the output currents, we define an augmented state vector, i.e., $\vec{x}_{odq} = [\delta {x}^{\top}_{odq} \; {x}^{\top}_{odq}]^{\top}$, and an input rate vector, i.e, $\delta u_{odq} =[\delta u^o_{d}\; \delta u^o_{q}]^{\top}$. 
Then, the augmented model of \eqref{eq:ss_output_current_dq} yields
\begin{equation}
\label{eq:ss_output_current_augmented_dqz}
    \vec{x}^{+}_{odq} = \textbf{A}^o \vec{x}_{odq} + \textbf{B}^o \delta u_{odq}
\end{equation}
with $\textbf{A}^o = \left(\vec{\textbf{A}}^{o} +\vec{\textbf{A}}_{\Delta}\right)$, $\textbf{B}^o=\left(\vec{\textbf{B}}^{o} +\vec{\textbf{B}}_{\Delta}\right)$, 
\begin{equation*}
\label{eq:ss_output_current_augmented_matrix_dqz}
         \vec{\textbf{A}}^{o} {=} \begin{bmatrix}
        \textbf{A}^o_0 & 0\\
        \textbf{C} \textbf{A}^o_0  &I
    \end{bmatrix} , \vec{\textbf{B}}^{o}{=} \begin{bmatrix}
        \textbf{B}^o_0\\
       \textbf{C} \textbf{B}^o_0
    \end{bmatrix},  
    \vec{\textbf{A}}_{\Delta} {=} \begin{bmatrix}
        \textbf{A}_{\Delta}\\
        \textbf{C} \textbf{A}_{\Delta}
    \end{bmatrix} , \vec{\textbf{B}}_{\Delta}{=} \begin{bmatrix}
        \textbf{B}_{\Delta}\\
       \textbf{C} \textbf{B}_{\Delta}
    \end{bmatrix},
\end{equation*}
where $\textbf{C} = I$, $\delta {x}_{odq}(k) = {x}_{odq}(k) - {x}_{odq}(k-1) $, $\delta u_{odq}(k) = u_{odq}(k) - u_{odq}(k-1) $. 
Similarly, we derive the augmented model of \eqref{eq:ss_circulating_current_dq} defining an augmented state vector, i.e., $\vec{x}_{cdq} = [\delta {x}^{\top}_{cdq} \; {x}^{\top}_{cdq}]^{\top}$, and an  input rate vector,i.e, $\delta u_{cdq} =[\delta u^c_{d}\; \delta u^c_{q}]^{\top}$. 
Then, its augmented model yields
\begin{equation}
\label{eq:ss_circulating_current_augmented_dqz}
    \vec{x}^{+}_{cdq} = \textbf{A}^c \vec{x}_{cdq} + \textbf{B}^c \delta u_{cdq},
\end{equation}
with $\textbf{A}^c = \left(\vec{\textbf{A}}^{c} +\vec{\textbf{A}}_{\Delta}\right)$, $\textbf{B}^c=\left(\vec{\textbf{B}}^{c} +\vec{\textbf{B}}_{\Delta}\right)$,
\begin{equation*}
\label{eq:ss_circulating_current_augmented_matrix_dqz}
         \vec{\textbf{A}}^{c} {=} \begin{bmatrix}
        \textbf{A}^c_0 & 0\\
        \textbf{C} \textbf{A}^c_0  &I
    \end{bmatrix} , \vec{\textbf{B}}^{c}{=} \begin{bmatrix}
        \textbf{B}^c_0\\
       \textbf{C} \textbf{B}^c_0
    \end{bmatrix},  
    \vec{\textbf{A}}_{\Delta} {=} \begin{bmatrix}
        \textbf{A}_{\Delta}\\
        \textbf{C} \textbf{A}_{\Delta}
    \end{bmatrix} , \vec{\textbf{B}}_{\Delta}{=} \begin{bmatrix}
        \textbf{B}_{\Delta}\\
       \textbf{C} \textbf{B}_{\Delta}
    \end{bmatrix}.
\end{equation*}

\subsection{MMC Control}
The MMC control system follows a hierarchical architecture composed of three levels. The upper level, referred to as the outer control loop, uses power references obtained from the optimal power flow (OPF) and generates current references $(i^{o*}_{dq}, i^{o*}_{dq})$ to ensure the desired AC-to-DC power transfer. 
The inner control loop then uses $i^{o*}_{dq}, i^{c*}_{dq}, i^{o}_{dq},$ and $i^{c}_{dq}$ to compute the corresponding voltage references $u^{o*}_{dq}$ and $u^{c*}_{dq}$. The circulating current reference is set to zero, such that  $i^{c*}_{dq} =0$.
At the lowest level, $u^{o*}_{dq}$ and $u^{c*}_{dq}$ are converted into modulation indices and switching signals for the submodule gates. This paper focuses on the inner control loop.
\begin{figure}[htp]
    \centering
    \includegraphics[width=\linewidth]{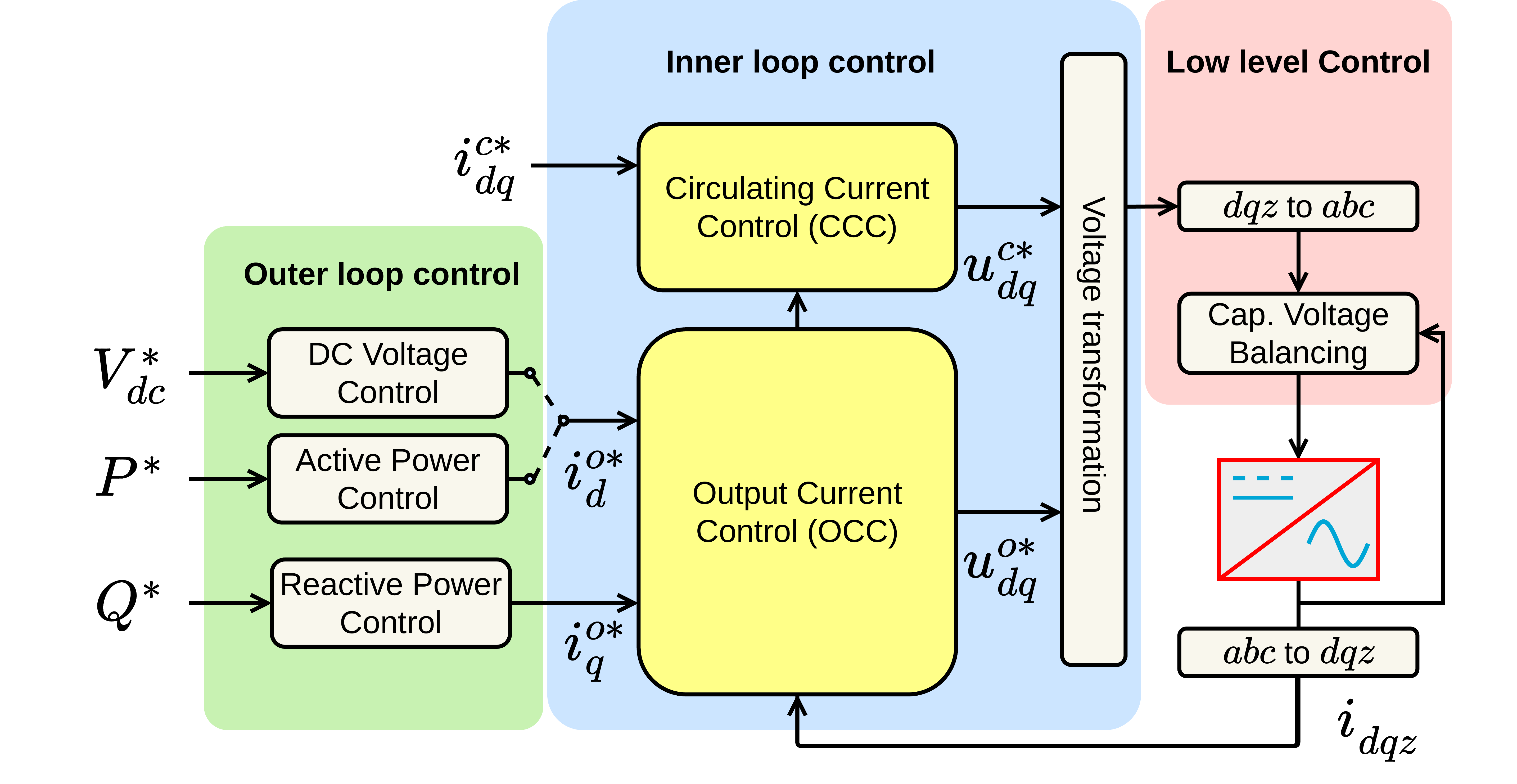}
    \caption{Generalized MMC Control Architecture.}
    \label{fig:SEST26_Control_Diagram}
\end{figure}

Beside regulating $i^o_{dq}$ and  $i^c_{dq}$, the inner control loop needs to avoid creating currents and voltage references that violate operation limits.
As an alternative to conventional solutions~\cite{Du2018MMCBook}, we use an architecture based on static state-feedback and feedforward gains, as shown in Fig.~\ref{fig:SEST26_Control_Architecture}.

\begin{figure} [htp]
    \centering
    \includegraphics[width=\linewidth]{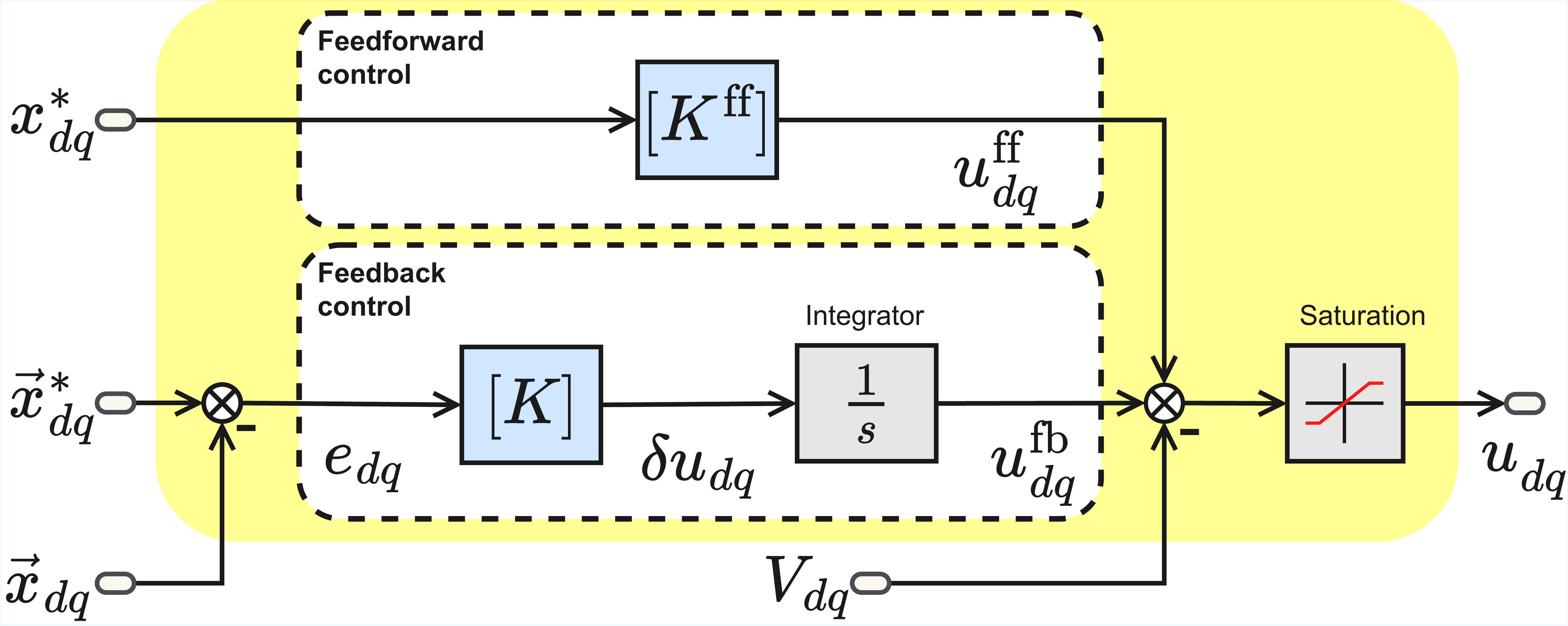}
    \caption{Generalized architecture of both inner loop controllers based on state feedback $(K)$ and feedforward $(K^{\text{ff}})$ gains.}
    \label{fig:SEST26_Control_Architecture}
\end{figure}

The feedforward term $u^{\text{ff}}$ is computed as follow
\begin{equation}
    u^{\text{ff}} =  K^{\text{ff}}x^{*}_{dq} \quad \text{with} \quad K^{\text{ff}}:=-\textbf{B}^{-1}_0\left(\textbf{I}-\textbf{A}_0\right).
\end{equation}
It yields steady-state voltage references for the nominal current dynamics.  
The feedback stage leverages the error of the augmented state vector, i.e, $e_{dq} {:=}  \vec{x}^{*}_{dq}{-}\vec{x}^{}_{dq}$, the static gain $K$ and an integrator block to recover the control input. Finally, a saturation block avoids over-voltage references.

\section{ Robust Current Regulator Design}
\label{lab:Methodology}
This section addresses the synthesis of $K$, i.e., the robust current regulator (RCR). 
First, we define performance constraints, which are denoted as:  
\begin{subequations}
    \label{eq:constraintDefinitionLemma}
        \begin{align}
            \label{eq:constraintDefinitionLemma_x}
            \mathbb{X} &:= \{x\in\mathbb{R}^{n_x}: g_{t_x}^\top x \leq 1, \forall t_x \in \{1,{...},s\}\},\\
            \label{eq:constraintDefinitionLemma_u}
            \mathbb{U} &:= \{ u\in\mathbb{R}^{n_u}: h_{t_u}^\top  u \leq 1, \forall t_u \in \{1,{...},l\}\},
        \end{align}
\end{subequations}
where $x$ and $u$ are generic state and input vectors, $\mathbb{X}$ and $\mathbb{U}$ are classic box constraints that relate to current and voltage reference constraints, respectively. 
The proposed controller synthesis leverages Lyapunov stability criteria to formulate the following optimization problem:
\begin{subequations}
\label{eq:optimzation_problem}
    \begin{align} 
    \label{eq:optimzation_problem_min}
        \min_{P,K}& \quad -\operatorname{det}(P)^{1/n_x} \\
                 \label{eq:optimzation_problem_lyapunov}
                  \text{s.t.}&\;\; \forall \rho \in \mathbb{P}, \quad (\textbf{A}+\textbf{B}K)^{\top}P(\textbf{A}+\textbf{B}K)-P \prec 0,\\
                  \label{eq:optimzation_problem_const}
                  &\;\; \forall x \in \mathbb{S}, \quad  x \in \mathbb{X}, \; \text{and} \;  u \in \mathbb{U}, \; \\
                  \label{eq:optimzation_problem_const_1}
                  &\;\; P \succ 0
    \end{align}
\end{subequations}
where $(\textbf{A},\textbf{B})$ is a pair of state-space matrices, in this case $(\textbf{A}^o,\textbf{B}^o)$ or $(\textbf{A}^c,\textbf{B}^c)$ depending of the control problem, $K \in \mathbb{R}^{n_u\times n_x}$ is the static gain, $P \in \mathbb{R}^{n_x\times n_x}$ defines a suitable quadratic candidate Lyapunov function, $\mathbb{S}$ is a domain of attraction such that
\begin{equation}
    \mathbb{S}:=\left\{x\in \mathbb{R}^n: x^{\top}Px \leq 1  \; \text{and} \; Kx \in \mathbb{U} \right\} \subseteq \mathbb{X}.
\end{equation}
We  consider $x := e_{dq} \in \mathbb{R}^{4}$, and $u := \delta u_{dq} \in \mathbb{R}^{2}$, where $e_{dq}$ and $\delta u_{dq}$ refer to the output or circulating currents accordingly. 
\textcolor{black}{Enforcing \eqref{eq:optimzation_problem_lyapunov}–\eqref{eq:optimzation_problem_const_1} ensures the fastest closed-loop response that is constraint-admissible across all state-space realizations, without using tuned hyperparameters such as $Q$ and $R$.}
\textcolor{black}{
\begin{lemma}
\label{lmm:main_lemma_paper}
    Consider the uncertain set $\mathbb{P}$ defined by a convex hull, such that
    \begin{equation}
        \label{eq:polytope_defintion}
        \mathbb{P} =  \text{conv} \{ v^1, v^2, \ldots, v^p\},
    \end{equation}
    where $p \in \mathbb{N}_{>0}$, $v^j$  are the $j$th vertex of the polytope $\mathbb{P}$ and $\text{conv} \{ \cdot\}$ denotes the operation of taking the convex hull of the argument.
    Let define matrices of the state-space realizations of \eqref{eq:ss_output_current_augmented_dqz} and \eqref{eq:ss_circulating_current_augmented_dqz} as affine matrix functions of $v^j$,  such that 
    \begin{equation}
    \label{eq:matrix_function_definition}
        \textbf{A}(v^j): \mathbb{R}^{n_{\rho}}\rightarrow \mathbb{R}^{n_x \times n_x}, \, \text{and} \, \textbf{B}(v^j): \mathbb{R}^{n_{\rho}}\rightarrow \mathbb{R}^{n_x \times n_u}
    \end{equation}
    with
    \begin{equation}
        \textbf{A}(v^j) = \textbf{A}_0+\textbf{A}_{\Delta}(v^j) \quad \text{and} \quad \textbf{B}(v^j) = \textbf{B}_0+\textbf{B}_{\Delta}(v^j)
    \end{equation}
    where $(\textbf{A}_0, \textbf{B}_0)$ describes the nominal dynamics and $(\textbf{A}_{\Delta}, \textbf{B}_{\Delta})$ are uncertain matrices filled with $v^j$ as in \eqref{eq:coefficient_unknown}. 
    Given \eqref{eq:constraintDefinitionLemma}, \eqref{eq:ss_output_current_augmented_dqz} and \eqref{eq:ss_circulating_current_augmented_dqz}, inequalities \eqref{eq:optimzation_problem_lyapunov}-\eqref{eq:optimzation_problem_const_1} hold, if there exists $Y$ and  $Z =Z^{\top} \succ0$, such that
    \begin{subequations}
    \label{eq:constrained_synthesis}
        \begin{align}
        \label{eq:constrained_synthesis_1_lyapunov}
        \begin{aligned}
            &\forall j \in \{1, ..., p\},\\
            &\begin{bmatrix}
            Z & (\textbf{A}(v^j)Z+\textbf{B}(v^j)Y)^{\top} \\
            (\textbf{A}(v^j)Z+\textbf{B}(v^j)Y) & Z
            \end{bmatrix} \succ 0, 
        \end{aligned}&  \\
        \label{eq:constrained_synthesis_2_State}
            \forall t_x \in\left\{1, \ldots, s\right\}, \quad \begin{bmatrix}
                Z & \left(Z g_{t_x}\right)^{\top} \\
                \left(Z g_{t_x}\right) & 1
            \end{bmatrix}  \succcurlyeq 0,& \\
            \label{eq:constrained_synthesis_3_Input}
            \forall t_u \in\left\{1, \ldots, l\right\}, \quad \begin{bmatrix}
                 Z & \left(Y h_{t_u}\right)^{\top} \\
                \left(Y h_{t_u}\right) & 1
            \end{bmatrix}\succcurlyeq 0,& 
        \end{align}        
    \end{subequations}
    where $P = Z^{-1}$, $K =YZ^{-1}$, $g_{t_x}$ and $h_{t_u}$ are defined as \eqref{eq:constraintDefinitionLemma}, and $p$ is the number of vertices of $\mathbb{P}$.
\end{lemma}
}
\begin{IEEEproof}
\textcolor{black}{
Substituting \eqref{eq:matrix_function_definition} in \eqref{eq:optimzation_problem_lyapunov} and applying Schur complement, \eqref{eq:constrained_synthesis_1_lyapunov} yields
    \begin{equation}
    \label{eq:proof_eq_1}
        \begin{bmatrix}
           Z & (\textbf{A}(v^j)Z+\textbf{B}(v^j)Y)^{\top} \\
            (\textbf{A}(v^j)Z+\textbf{B}(v^j)Y) & Z
            \end{bmatrix} \succ 0
    \end{equation}
    with $Z = P^{-1}$. 
    Derivation of \eqref{eq:proof_eq_1} are well established, see \cite{apkarian2000parameterized}.
    $\mathbb{P}$ is a convex polytope and $\textbf{A}(\cdot)$ and $\textbf{B}(\cdot)$ are affine matrix functions. Based on multi-convexity \cite[Collorary 2.3]{apkarian2000parameterized}, \eqref{eq:proof_eq_1} holds for all $\rho \in \mathbb{P}$ if \eqref{eq:proof_eq_1} holds for all $v^j$.\\
    Applying Schur complement, \eqref{eq:constrained_synthesis_2_State} and \eqref{eq:constrained_synthesis_3_Input} yield
    \begin{subequations}
    \label{eq:proof_eq_2}
    \begin{align}
    \label{eq:proof_eq_2_1}
    \sqrt{g_{t_x}^\top P^{-1} g_{t_x}} &\leq 1, \quad \forall t_x {\in }\{1,2,{...},s\},\\
    \label{eq:proof_eq_2_3}
        \sqrt{h_{t_u}^\top K P^{-1} K^\top h_{t_u}} &\leq 1, \quad \forall t_x {\in }\{1,2,{...},l\},
    \end{align}
    \end{subequations}
    respectively. 
    Let $P^{1/2}$ be a symmetric square root, it holds
    \begin{equation}
    \label{eq:proof_eq_3}
         g_{t_x}^\top x {=} \left(P^{-1/2} g_{t_x}\right)^{\top} \left(P^{1/2} x\right) \leq \left\|P^{-1/2} g_{t_x}\right\| \, \left\|P^{1/2} x\right\|,
    \end{equation}
    for all $ t_x {\in }\{1,2,{...},s\}$ by Cauchy–Schwarz inequality.
    From \eqref{eq:proof_eq_3}, it follows
    \begin{equation}
        g_{t_x}^{\top} x \le \sqrt{g_{t_x}^{\top} P^{-1} g_{t_x}} \,\sqrt{x^{\top} P x} \le\sqrt{g_{t_x}^{\top} P^{-1} g_{t_x}} \leq 1.
    \end{equation}
    So for any $x \in \mathbb{S}$, if $g_{t_x}^{\top} P^{-1} g_{t_x} \le 1$, then $g_{t_x}^{\top} x \le 1$. Therefore, $\mathbb{S} \subseteq \mathbb{X}$. This is a standard proof of ellipsoidal inclusion, see \cite{apkarian2000parameterized}. The same reasoning applies to prove that with $u = Kx$ and $\mathbb{U}$  as \eqref{eq:constraintDefinitionLemma_u}, it follows that $K\mathbb{S} \subseteq \mathbb{U}$ for all $x\in \mathbb{S}$.}
\end{IEEEproof}

\begin{theorem}
\label{thm:the_theorem}
    Given a state-space realization such as \eqref{eq:ss_output_current_augmented_dqz} or \eqref{eq:ss_circulating_current_augmented_dqz}, and its polytopic constraints $\mathbb{P}$, $\mathbb{X}$ and $\mathbb{U}$. 
    The solution of the optimization problem 
    \begin{equation}
        \begin{aligned}
             \min_{P,K}& \quad -\operatorname{det}(P)^{1/n_x}\\
             &\text{s.t.} \quad \eqref{eq:constrained_synthesis}
        \end{aligned}
    \end{equation}
    yields an optimal RCR.  
\end{theorem}
\begin{IEEEproof}
    \textcolor{black}{The objective function \eqref{eq:optimzation_problem_min} is a convex function that maximizes $\mathbb{S}$ by minimizing the eigenvalues of $P$ such that $K$ ensure constraint-admissible inputs for the largest region possible.
    The rest of the proof follows from Lemma~\ref{lmm:main_lemma_paper}.}
\end{IEEEproof}

Theorem~\eqref{eq:optimzation_problem} can be solved using semi-definite programming (SDP) solver MOSEK, since \eqref{eq:optimzation_problem_min} is a convex objective function and~\eqref{eq:optimzation_problem_lyapunov}-\eqref{eq:optimzation_problem_const} can be expressed as LMIs.

\section{Results}
\label{lab:Results}
To assess the RCR effectiveness, we consider the two-terminal HVDC benchmark, i.e., CIGRE B4.57 DCS1 \cite{Cigre2020dc,Vrana2013Cigre},  implemented in RTDS\textsuperscript{\textregistered} (see its simplified diagram in Figure~\ref{fig:MTDC_Diagrram_Two_Terminals}).
Full description of the model is found in \cite{Shetgaonkar2024_MTDCMOdel}, but
MMC parameters are summarized in Table~\ref{tab:Tennet_MTDC_2GW_parameters}.

\begin{figure}[!htp]
    \centering
    \includegraphics[width=\linewidth]{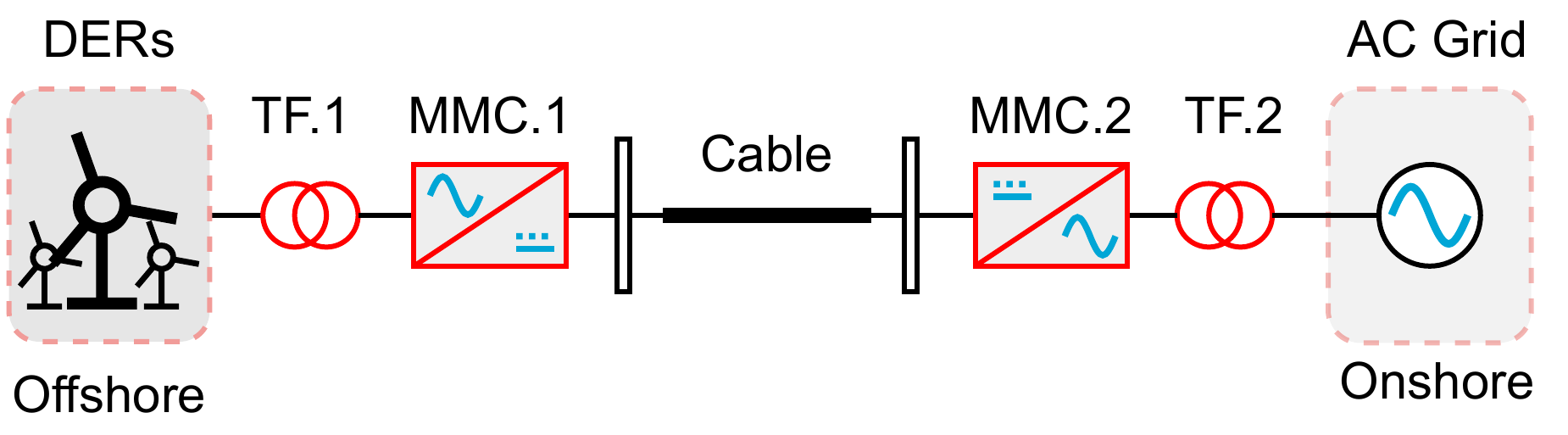}
    \caption{Simplified diagram  of the two terminal HVDC transmission system.}
    \label{fig:MTDC_Diagrram_Two_Terminals}
\end{figure}

\begin{table}[htp]
    \centering
    \caption{\textcolor{black}{MTDC rated parameters.}} 
    \begin{tabular}{lclc}
    \hline \multicolumn{1}{c}{\textbf{Parameter}} & \multicolumn{1}{c}{\textbf{Value}} & \multicolumn{1}{c}{\textbf{Parameter}} & \multicolumn{1}{c}{\textbf{Value}} \\
    \hline Rated power & 0.8 GVA & Rated frequency & 50 Hz \\
    MMCs AC voltage & 220 kV & MMCs DC voltage & 400 kV \\
    Onshore AC voltage & 380 kV & Offshore AC voltage & 145 kV \\
    Number of SMs & 200 & SM capacitance & 10 mF \\
    SM switching voltage & 2.1 kV & SM switching current & 1 kA \\
    Arm resistance & $0.15 \Omega$ & Arm inductance & 29 mH \\
    TFs resistance & 0.006 pu & TFs inductance & 0.18 pu \\
    \hline
    \end{tabular}
    \label{tab:Tennet_MTDC_2GW_parameters}
\end{table}

\textcolor{black}{For the RCR synthesis, it is desired that up to 1~p.u. step in the current reference produces control input variations below 0.2~p.u. Accordingly, the constraints are defined as:}
\begin{equation}
\label{eq:def_study_case}
\begin{aligned}
    \mathbb{X}&:= \{\vec{e}_{dq}\in \mathbb{R}^{4}: \|\vec{e}_{dq}\|_{\infty} \leq 1\}\quad \text{and} \\ \mathbb{U}&:= \{\delta{u}_{dq}\in \mathbb{R}^{2}: \|\delta{u}_{dq}\|_{\infty} \leq0.2\}.
\end{aligned}    
\end{equation}
The uncertainty matrices, for both OCC and CCC, are defined by
\begin{equation}
    \textbf{A}_{\Delta}= \begin{bmatrix}
        \rho_1 &\rho_2\\ \rho_3 &\rho_4
    \end{bmatrix} \quad \text{and} \quad \textbf{B}_{\Delta}= \begin{bmatrix}
        \rho_5 &0\\ 0 &\rho_6
    \end{bmatrix}
\end{equation}
with 
\begin{equation*}
\label{eq:scheduling_set}
    \begin{array}{lll}
         \rho_1 \in [-\frac{6}{10^2},\frac{6}{10^2}], &\rho_2 \in [-\frac{5}{10^3},\frac{5}{10^3}], &\rho_3 \in [-\frac{5}{10^3},\frac{5}{10^3}], \\ \rho_4 \in [-\frac{6}{10^2},\frac{6}{10^2}],
         &\rho_5 \in [-\frac{3}{10^{4}},\frac{3}{10^{4}}], &\rho_6 \in [-\frac{3}{10^{4}},\frac{3}{10^{4}}].
    \end{array}
\end{equation*}
Then, employing Theorem~\ref{thm:the_theorem}, the RCR yields 
\begin{equation}
    K_{\text{RCR}} {=} \left[\begin{array}{rrrr}
-1.8887 & 0.0115 & -0.0407 & -0.0012 \\
-0.0115 & -1.8887 & 0.0012 & -0.0407
\end{array}\right]
\end{equation}
Transforming \eqref{eq:def_study_case} into \eqref{eq:constraintDefinitionLemma} is a standard procedure that is not explained due to the page limitations.  
To compare the RCR to the state-of-the-art, we consider the optimal robust current regulator (OCR) from \cite{Ayari2017_RobustHVDC} to compute $K_{\text{OCR}}$. 
We select 
$Q = 10^{4}\textbf{I}$ and $R = 10^{-4}\textbf{I}$, which yields
\begin{equation}
    K_{\text{OCR}} = \left[ \begin{array}{rlll}
-4.0878 & 0.3678 & -0.4614 & -1.8765 \\
0.8199 & 1.1786 & -0.1319 & -0.5841
\end{array}\right]
\end{equation}

\subsection{RCR Verification via MATLAB Simulations}
We verify the RCR properties by simulating in MATLAB the inner controller in closed-loop with DT uncertain models \eqref{eq:ss_output_current_dq} and \eqref{eq:ss_circulating_current_dq}. 
Figure~\ref{fig:sest26_comparison} presents $N_{\text{unc}}$ uncertain realizations of the output current closed-loop response when both controllers are used. 
We quantify the performance of the controller design by assessing the average mean absolute error (MAE) of the uncertain realization to the nominal response, such that:
\begin{equation}
    \text{KPI} = \frac{1}{N_{\text{unc}}}\sum^{N_{\text{unc}}}_i\sum^{T}_k {\frac{1}{T}\left\|{x^{\text{nom}}_{dq}(k)-x^{i_\text{unc}}_{dq}(k)}\right\|}
\end{equation}
where $T$ is the amount of sample during the simulation, $x^{\text{nom}}_{dq}$ is the nominal state trajectories,  $x^{i_\text{unc}}_{dq}$ are the state trajectories for $N_{\text{unc}}$ different uncertain realizations. 
Assessing the nominal responses, the RCR achieves a settling time of 4.0~ms, whereas the OCR settles at 4.3~ms. 
In both cases, robustness is guaranteed; however, the RCR provides, on average, a faster response and maintains behavior closer to the nominal case across all uncertainty realizations with $\text{KPI}_{\text{RCR}} =  445.9716$ and $\text{KPI}_{\text{OCR}} =  714.3589$.

\begin{figure}[htbp]
    \centering
    \includegraphics[width=\linewidth]{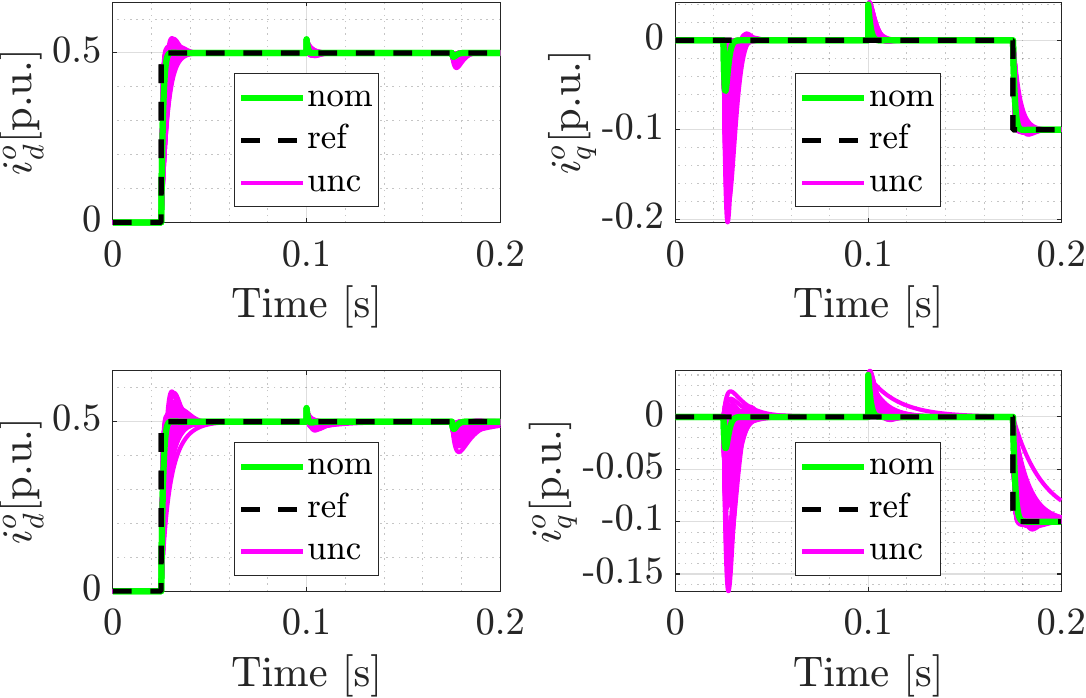}
    \caption{MMC closed-loop response output current considering $N_{\text{unc}}=200$ uncertain realizations: (top-left) $i^{o}_d$ under RCR, (top-right) $i^{o}_q$ under RCR, (bottom-left) $i^{o}_d$ under OCR \cite{Ayari2017_RobustHVDC}, (bottom-right) $i^{o}_q$ under OCR \cite{Ayari2017_RobustHVDC}.}
    \label{fig:sest26_comparison}
\end{figure}

\subsection{Controllers validation via RTDS\textsuperscript{\textregistered} simulations}
Next, we present the results obtained via RTDS simulations. 
RTDS\textsuperscript{\textregistered} real-time simulator considers average nonlinear dynamics, noise, and disturbance, while allowing us to accurately recreate faults and power demand fluctuations. 
Figure~\ref{fig:SEST26_POWER_Vg_cropped} shows the AC transformer voltage of the onshore MMC under a symmetrical three-phase line-to-ground (LLLG) fault.

\begin{figure}[!htp]
    \centering
    \includegraphics[width=\linewidth]{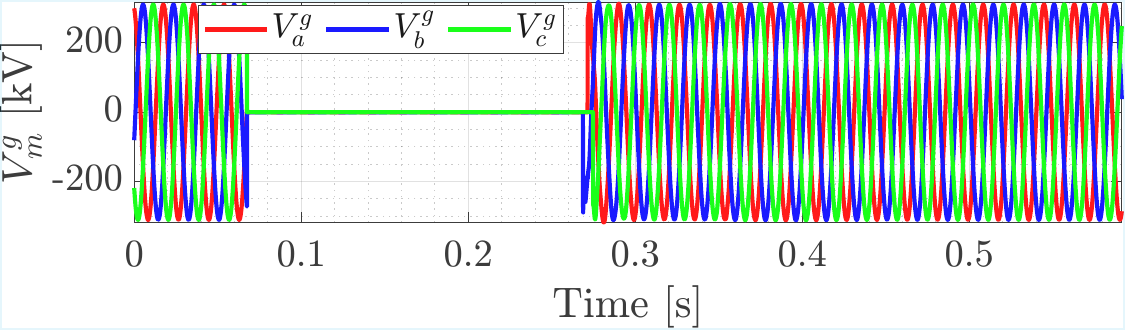}
    \caption{AC transformer voltage trajectory, in kV, when phase-to-ground fault occurs.}
    \label{fig:SEST26_POWER_Vg_cropped}
\end{figure}

Figure~\ref{fig:SEST26_POWER_Ig} illustrates the output current response under RCR and OCR control. 
Both controllers successfully clear the fault and restore power delivery to 200~MW after AC transformer recovery. 
During fault clearance, the RCR limits the overcurrent to 2.25033~kA, whereas the OCR reaches 2.51458~kA; both exhibit comparable settling times. Upon returning to the pre-fault operating point, the RCR again limits the peak current to -6.51704~kA, compared to -6.74151~kA for the OCR.

\begin{figure}[!htp]
    \centering
    \includegraphics[width=\linewidth]{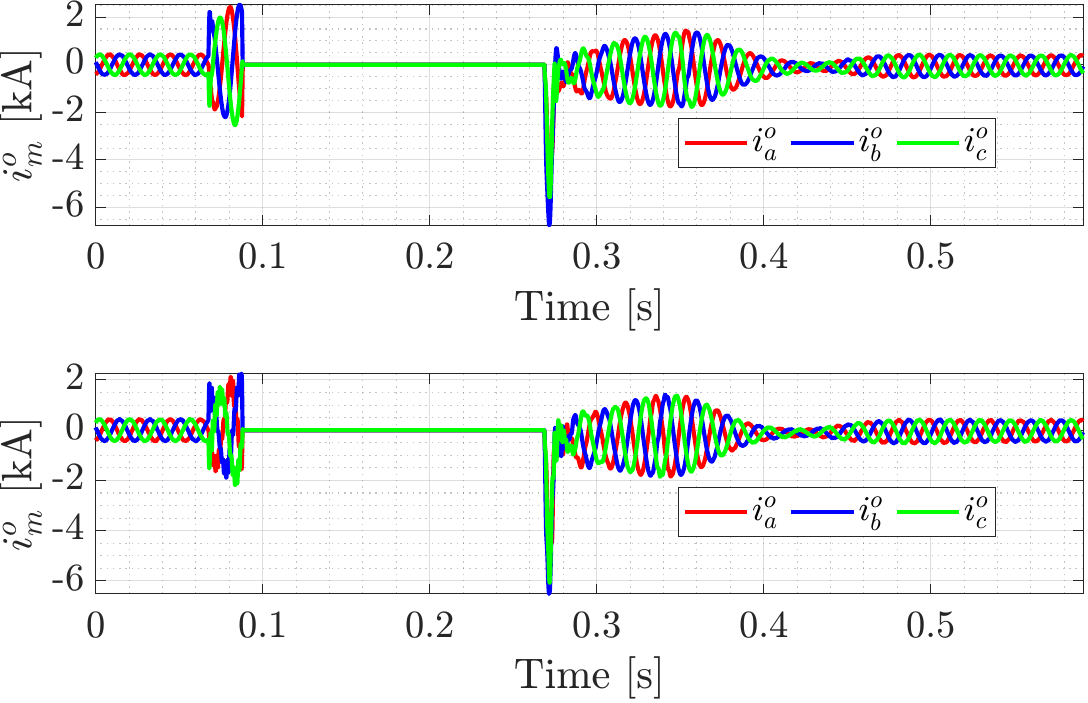}
    \caption{Output current response in [kA]: (top) OCR and (bottom) RCR.}
    \label{fig:SEST26_POWER_Ig}
\end{figure}

Figure~\ref{fig:SEST26_POWER_PQ} presents the active and reactive power measured at the MMC AC terminals. 
Both controllers achieve similar settling times, i.e., ${\approx}0.23$s. 
During post-fault recovery, the RCR creates a lower power peak of 5.628MW.
We quantify the power excess via the MAE, where $\mathrm{MAE}(\text{PMMC},\text{QMMC}) = (77.84~\mathrm{MW},8.64~\mathrm{MVAR})$ for RCR and $=(79.65~\mathrm{MW},25.06~\mathrm{MVAR})$ for OCR.
\begin{figure}[!htp]
    \centering
    \includegraphics[width=\linewidth]{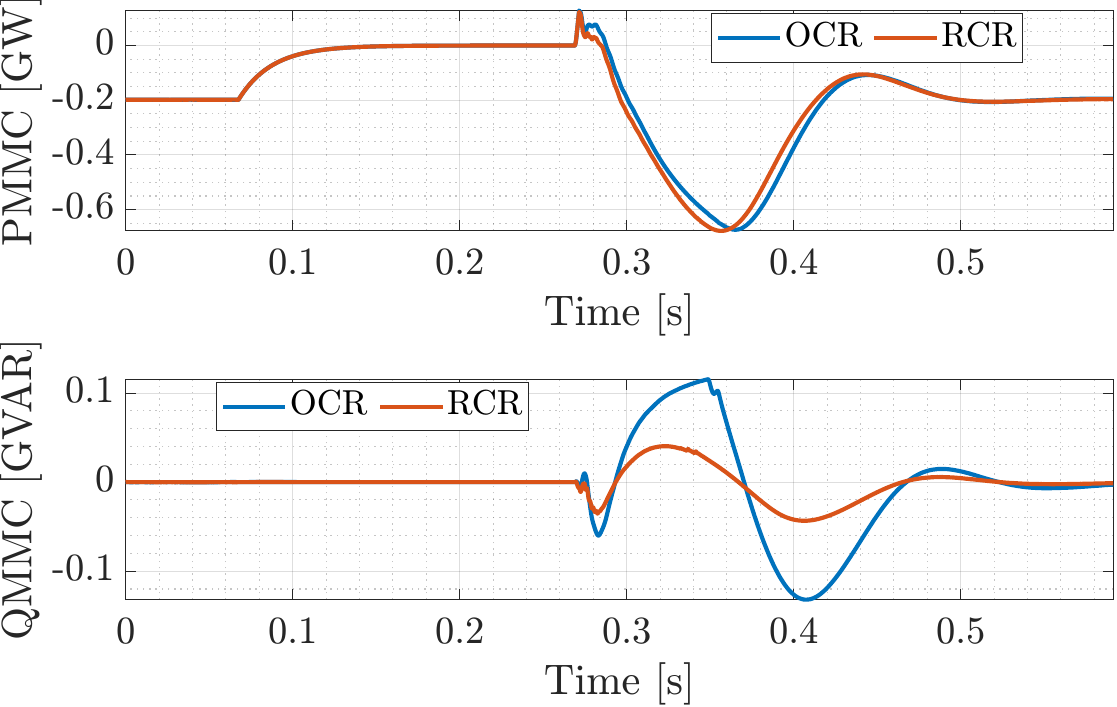}
    \caption{MMC active power (top) and reactive power (bottom). }
    \label{fig:SEST26_POWER_PQ}
\end{figure}

The performance improvement is further illustrated in Figures~\ref{fig:SEST26_POWER_OCC_DQ} and~\ref{fig:SEST26_POWER_CCC_DQ}, which present the output and circulating currents in the $dq$ reference frame.
In the $d$-axis, the RCR reduces the peak overcurrent by 0.22785~p.u. compared to the OCR, while maintaining a comparable settling time. In the $q$-axis, it achieves a 0.18852~p.u. reduction in peak magnitude and a 5\% faster transient response.
A similar trend is observed in the circulating current dynamics. The RCR provides consistently improved regulation accuracy, as quantified by the MAE: $\mathrm{MAE}(i^{c}_d,i^{c}_q) = (0.0034, 0.0138)$~p.u. for RCR and $(0.0113, 0.0261)$~p.u. for OCR.

\begin{figure}[!htp]
    \centering
    \includegraphics[width=\linewidth]{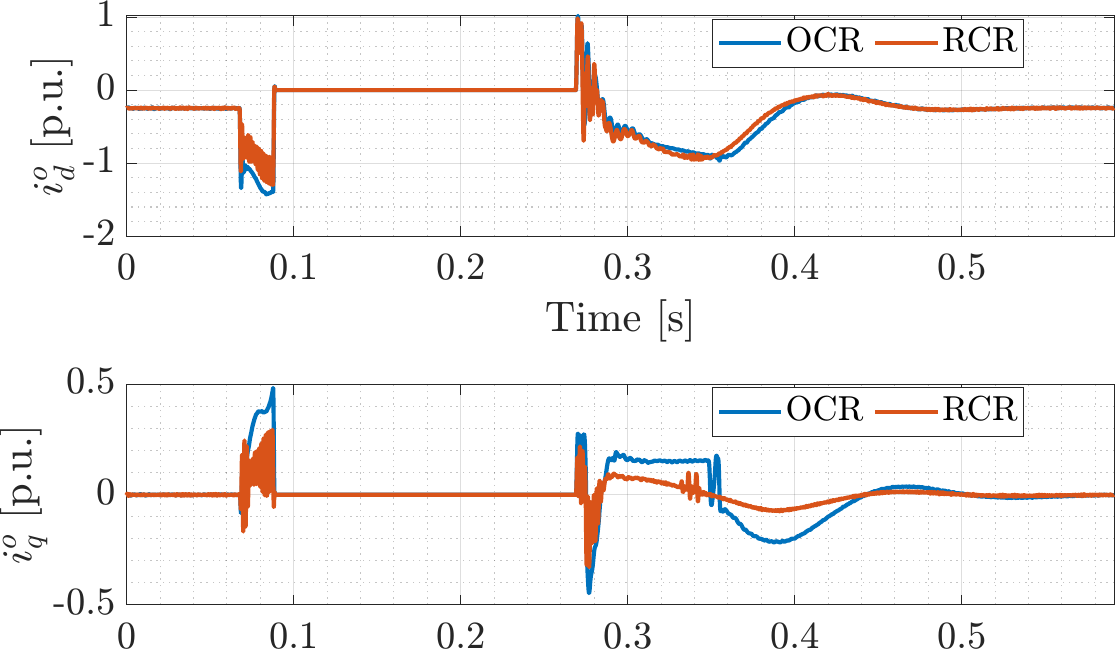}
    \caption{$dq$ output current response in [p.u.]: (top) OCR and (bottom) RCR.}
    \label{fig:SEST26_POWER_OCC_DQ}
\end{figure}

\begin{figure}[!htp]
    \centering
    \includegraphics[width=\linewidth]{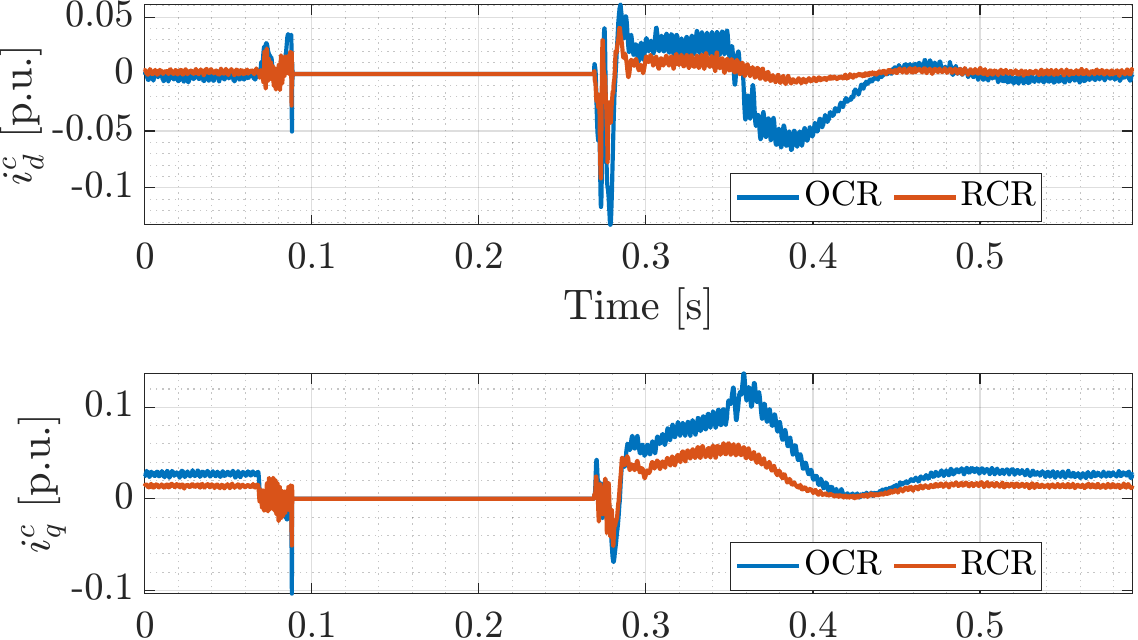}
    \caption{$dq$ circulating current response in [p.u.]: (top) OCR and (bottom) RCR.}
    \label{fig:SEST26_POWER_CCC_DQ}
\end{figure}

\section{Conclusion}
In this paper, we present a robust current regulator that leverage  Lyapunov inequalities to ensure robustness while pushing controller performance. 
As result, the proposed controller is capable to provide the faster transient within the safe operation region for all possible uncertain realization. 
To show the efficacy of the proposed method, we compare it with respect to the OCR by \cite{Ayari2017_RobustHVDC} showing faster transient and lower tracking error. 
Furthermore, the use of the proposed method can be extended to any linear state-space realization.

\bibliographystyle{IEEEtran}
\bibliography{mijnbib}

\end{document}